\documentclass{amsart}
\usepackage[foot]{amsaddr} 

\usepackage{geometry}
\usepackage[numbers]{natbib}
\usepackage{doi}

\usepackage{bbm}
\usepackage{mathtools}
\usepackage{multicol}
\usepackage{tikz}
\usepackage{verbatim}
\usepackage{lipsum}
\usepackage{amsfonts}
\usepackage{graphicx}
\usepackage{epstopdf}
\usepackage{algorithmic}
\usepackage{amsmath}
\usepackage{amssymb}
\usepackage{multirow}
\usepackage{amsthm}
\usepackage{cleveref}

\usetikzlibrary{matrix}
\usetikzlibrary{positioning} 
\usetikzlibrary{decorations.pathmorphing}
\usetikzlibrary{shapes.geometric, arrows.meta}

\newcommand{\Var}{\mathrm{Var}} 
\newcommand{\Cov}{\mathrm{Cov}} 

\theoremstyle{plain}
\newtheorem{theorem}{Theorem}[section]
\newtheorem{lemma}[theorem]{Lemma}

\keywords{bounds on probability; effective sample size; meta-analysis; random effects model; uncertainty quantification}

\begin{document}

\title[Iterative importance sampling with Markov chain Monte Carlo sampling]{Iterative importance sampling with Markov chain Monte Carlo sampling in robust Bayesian analysis}

\author{Ivette Raices Cruz$^1$}
\address{$^1$Centre for Environmental and Climate Science, Lund University, Lund, Sweden}
\email{ivette.raices\_cruz@cec.lu.se}

\author{Johan Lindstr{\"o}m$^2$}
\address{$^2$Centre for Mathematical Sciences, Lund University, Lund, Sweden}
\email{johan.lindstrom@matstat.lu.se}

\author{Matthias C.M. Troffaes$^3$}
\address{$^3$Durham University, Department of Mathematical Sciences, UK}
\email{matthias.troffaes@durham.ac.uk}

\author{Ullrika Sahlin$^1$}
\email{ullrika.sahlin@cec.lu.se}

\begin{abstract}
Bayesian inference under a set of priors, called robust Bayesian analysis, allows for estimation of parameters within a model and quantification of epistemic uncertainty in quantities of interest by bounded (or imprecise) probability. 
Iterative importance sampling can be used to estimate bounds
on the quantity of interest by optimizing over the set of priors. 
A method for iterative importance sampling when the robust Bayesian inference rely on Markov chain Monte Carlo (MCMC) sampling is proposed.
To accommodate the MCMC sampling in iterative importance sampling, a new expression for the effective sample size of the importance sampling is derived, which accounts for the correlation in the MCMC samples. 
To illustrate the proposed method for robust Bayesian analysis, iterative importance sampling with MCMC sampling is applied to estimate the lower bound of the overall effect in a previously published meta-analysis with a random effects model. 
The performance of the method compared to a grid search method and under different degrees of prior-data conflict is also explored.
\end{abstract}

\maketitle

\section{Introduction}

Bayesian analysis quantifies uncertainty by precise probability derived from a prior (subjective) distribution for parameters and a 
likelihood for data given parameters \citep{Gelman_2013}. 
Whereas statistical Bayesian inference usually uses non-informative priors as default, there are exceptions motivating the use of informative priors to reduce complexity \citep{Simpson_2017}. In a decision context where one wants to use the best possible knowledge, informative priors are
useful, or even needed, to integrate data with expert knowledge. Specifying a precise informative prior may be difficult, in particular for a model with many parameters or when experts disagree \citep{Rinderknecht_2012, Insua_2000}.

Robust Bayesian analysis is a way to consider the impact of the choice of prior on uncertainty 
in relevant quantities. The impact of different priors in Bayesian inference is important to evaluate for two reasons. First, it is common that more than one prior probability distribution could reasonably be chosen for the problem at hand. Second, when information in data is weak (e.g. for small sample sizes), the choice of prior could matter a lot for the final outcome of an analysis.
Robust Bayesian analysis has been used for sensitivity analysis towards the choice of prior \citep{Berger_1990}. 

A type of robust Bayesian analysis is to use sets of prior distributions or sets of likelihoods resulting in sets of posterior distributions. This can be seen as an extension of Bayesian inference which quantifies uncertainty by bounded (imprecise) probability
instead of precise probability \citep{Walley_1991}.  
In robust Bayesian analysis, one is often interested in estimating bounds on expectation.  For instance, the lower bound on expectation of a function $f$ with respect to a set of posterior distributions, $\mathcal{M}$, is expressed as
\begin{equation}
\underline{\mathrm{E}}(f) \coloneqq \displaystyle \inf_{p \in \mathcal{M}} \int f(x)p(x) dx. \label{LowerExpectation}
\end{equation}
Note that a set of posterior distributions is derived from a set of prior distributions.  

Robust Bayesian analysis using sets of priors has been developed in a closed analytic form for conjugate models \citep{Bernard_2005, Quaeghebeur_2005, Walley_1996}. 
In \citep{Wei_2017}, a range of posterior expectations 
are computed using a Monte Carlo method when considering uncertainty regarding the prior or likelihood.

Importance sampling has also been used to estimate bounds on expectations using independent samples drawn from arbitrary (e.g. not necessarily 
conjugate) models, as long as the posterior can be analytically evaluated up to a 
normalization constant \citep{Fetz_2017, Troffaes_2017, Troffaes_2018, Fetz_2018}.
In an iterative version of importance sampling, it has been suggested to iteratively change the sampling (also called proposal) distribution 
of importance sampling, in order to get an effective sample size (i.e. a measure of efficiency) as close as possible to the actual sample size \citep{Troffaes_2017, Troffaes_2018}. 
A small effective sample size means that the weights of importance sampling  are too imbalanced and thus might be unreliable.

In \citep{Troffaes_2018}, it is also suggested to use the posterior distribution directly as a sampling distribution where possible. 
The use of the posterior allows, in theory, for the effective sample size to be maximized 
across iterations.
For this reason, the effective sample size of the importance sampling estimator is used as the stopping criterion of iterative importance sampling in \citep{Troffaes_2018}. 
However, the effective sample size can be very poor if the sampling distribution is not 
carefully chosen, i.e. if the initial choice of posterior is far from the posterior that, say, minimizes the expectation of the quantity of interest. Further, the method can have issues with convergence across iterations as a result.
The required number of samples for accurate estimation using importance sampling has also been discussed in \citep{Agapiou_2017, Chatterjee_2018, Sanz_2018} by means of the Kullback Leibler divergence.

Markov chain Monte Carlo (MCMC) sampling is a
method for Bayesian inference which does not require a closed form of the posterior \citep{Brooks_2011, Gelman_2013}. MCMC sampling allows for inference of simple as well as complex models. So far, few robust Bayesian analysis have used MCMC sampling (see \citep{Vernon_2017} for an example). This is mainly due to limitations in existing methods, such as requirements on knowing the analytical form of the posteriors. The inability to use MCMC sampling severely restricts the use of robust Bayesian analysis on more complex models.

To use iterative importance sampling with MCMC samples there is a need to modify the effective sample size that is used in the stopping criterion in iterative importance sampling. In this paper, we combine iterative importance sampling with MCMC sampling by extending the method from \citep{Troffaes_2017, Troffaes_2018} to a wider range of models, specifically those requiring MCMC sampling. To accomplish this, we derive an expression for the effective sample size which accounts for correlated MCMC samples. 

In \citep{Liesenfeld_2008}, an efficient importance sampling is proposed for improving a MCMC algorithm. The efficient importance sampling consists of selecting a proposal distribution (given a density kernel) using a least squares problem and then using the proposed distribution in an independent Metropolis Hasting sampling. 
Moreover, in \citep{Llorente_2021}, a layered adaptive importance sampling algorithm is presented which combined MCMC algorithms with importance sampling and different strategies to re-use generated samples. The layered adaptive importance sampling algorithm generates samples using two layers. The upper layer generates samples using a MCMC algorithm which are later used in a multiple importance sampling scheme (lower layer) \citep{Llorente_2021}. A recycling layered adaptive importance sampling scheme is presented which re-uses the samples from the upper layer in the lower layer \citep{Llorente_2021}. This scheme is similar to the method proposed in this paper. 
However, a key difference between the \citep{Liesenfeld_2008} and \citep{Llorente_2021} papers and this paper is that we re-use the MCMC samples from one run by weighting with a different prior, and that we use an efficient sample size for determining when to re-run the MCMC sampler, which leads to fewer MCMC runs to cover the hyperparameter space. 

A robust Bayesian analysis allows for a quantification of uncertainty in quantities of interest which is robust to the choice of prior. It can also be useful for Bayesian inference when priors are given as sets. To demonstrate the proposed method, iterative importance sampling with MCMC sampling is applied to estimate the lower bound of the overall effect of biomanipulation of freshwater lakes in an already published meta-analysis \citep{Bernes_2015}. 

\section{Importance sampling}
\label{sec:IS}

Recall that the expectation of a function $f$ with respect to a \emph{target distribution}, $p$ is given by
\begin{equation}
\mu \coloneqq
\mathrm{E}_p \Bigl(f(x)\Bigr) =  \displaystyle \int f(x)p(x) dx,
\end{equation}
which can be approximated by the standard Monte Carlo estimator of $\mathrm{E}_p \Bigl(f(x)\Bigr)$ as 
\begin{equation}
\overline{\mu} \coloneqq  \frac{1}{N}\displaystyle \sum_{i=1}^N  f(X_i), 
\label{Eq:Exp_MC}
\end{equation}
where  $X_i \sim p$ are independent and identically distributed (i.i.d) samples.

Sometimes, it is difficult to draw samples directly from $p$ or there is a mismatch between $f(x)$ and $p(x)$ \citep{Owen_2013}. 
In this case, importance sampling could be applied to estimate the expectation by weighting samples drawn from a \emph{sampling distribution}, $q$, from which it is easier to generate samples
\citep{Egloff_2010}. This technique has been applied in Bayesian inference \citep{Gelman_2013} and in numerical integration \citep{Liu_2008, Owen_2013}.

The sampling density function $q$ must be such that $q(x) > 0$ whenever $p(x)f(x) \neq 0$. 
Sometimes $p$ is only known up to a normalization constant,  say only $p_u = cp$ is known, 
where $c > 0$ is an unknown constant.
The expectation of $f$ with respect to $p$ can be written as
\begin{align}
\mu \coloneqq
\mathrm{E}_p\Bigl(f (x)\Bigr)  &= \frac{\displaystyle \int f(x)p_u(x) dx}{\displaystyle \int p_u(x) dx} = \frac{\displaystyle \int f(x) w_p(x) q(x) dx}{\displaystyle \int w_p(x)q(x) dx} = \frac{\mathrm{E}_q\Bigl(f(x)w_p(x)\Bigr)}{\mathrm{E}_q\Bigl(w_p(x)\Bigr)}, 
\end{align}
where $w_p(x)\coloneqq \frac{p_u(x)}{q(x)}$.

This expectation can be estimated by \emph{self-normalized importance sampling}, which is defined as
\begin{equation}
\widetilde{\mu}
\coloneqq
\frac{\sum_{i=1}^N f(X_i)w_p(X_i)}{\sum_{i=1}^N w_p(X_i)}
\qquad\text{where }X_i\sim q \text{ (i.i.d)}.
\label{Eq:Exp_IS}
\end{equation}

In the following, we shall relax the assumption of independence.
At this point,
it suffices to point out that \cref{Eq:Exp_IS} is still
a valid estimator of $\mathrm{E}_p \Bigl(f(x) \Bigr)$
even if the $X_i$ are dependent.

A measure to assess the quality of the importance sampling estimator is the effective sample size, $\mathrm{ESS}$. It is defined by \citep{Kong_1992}, as the ratio of the variances of $\overline{\mu}$ and $\widetilde{\mu}$ 
estimators, in \cref{Eq:Exp_MC} and \cref{Eq:Exp_IS}, scaled to N:
\begin{align}
\mathrm{ESS}
\coloneqq
\frac{N\Var_p(\overline{\mu})}{\Var_q(\widetilde{\mu})}. \label{Eq:ESS}
\end{align}
The effective sample size represents the number of
standard Monte Carlo samples that are needed for both estimators to have the
same variance.

If the samples $X_i$ from  $q$ are i.i.d. then $\textrm{ESS}$ can be estimated by 
\begin{equation}
\mathrm{ESS_{IS}} \coloneqq
\frac{(\sum_{i=1}^N w_p(X_i))^2}{\sum_{i=1}^N w^2_p(X_i)},
\label{Eq:ESS_aprox}
\end{equation}
see \citep{Owen_2013}. 
However, this formula is not applicable for correlated samples, as would be the case if the $X_i$ are sampled
from $q$ using an MCMC algorithm.

\section{Effective sample size of importance sampling using MCMC}
\label{sec:ESS}

Here, we want to derive an estimate of the effective sample size as defined in \cref{Eq:ESS} 
when the $X_i$ are sampled from $q$ through MCMC. First, we derive an approximation of $\Var_q(\widetilde{\mu})$.

Importance sampling with MCMC was introduced by \citep{Hastings_1970} for variance reduction \citep{Bhattacharya_2008}. Hastings \citep{Hastings_1970}
suggested the following approximate of the denominator in \cref{Eq:ESS}:
\begin{align}
\Var_q \left (\widetilde{\mu} \right ) = \Var_q\left ( \frac{\overline{Y}}{\overline{Z}} \right )
& \approx 
\frac{%
  \Var_q(\overline{Y}-\mu\overline{Z})}{%
  \Bigl(\mathrm{E}_q(\overline{Z})\Bigr)^2%
},
\label{Eq:est_Var_1}
\end{align}
where
\begin{align}
\overline{Y} & \coloneqq \frac{1}{N}\sum_{i=1}^N f(X_i)w_p(X_i),   \\
\overline{Z} & \coloneqq \frac{1}{N}\sum_{i=1}^N w_p(X_i), 
\end{align}
and $\mu\coloneqq\mathrm{E}_p\Bigl(f(x)\Bigr)$ as defined earlier.

First, let us evaluate the denominator in \cref{Eq:est_Var_1}.
Note that
\begin{equation}
\label{eq:expqwpxi}
\mathrm{E}_q\Bigl(w_p(X_i)\Bigr)
=\int w_p(x)q(x)dx
=\int \frac{c p(x)}{q(x)}q(x)dx
=c, 
\end{equation}
so $\mathrm{E}_q(\overline{Z})=c$, and we can approximate the denominator via
\begin{align}
\label{eq:eqzsquaredapprox}
\Bigl(\mathrm{E}_q(\overline{Z})\Bigr)^2
&\approx
\left(\frac{1}{N}\sum_{i=1}^N w_p(X_i)\right)^2. 
\end{align}
The numerator in \cref{Eq:est_Var_1} is more tricky.
Let
\begin{equation}
g(X_i)\coloneqq (f(X_i)-\mu) w_p(X_i). \label{Eq:g_tilde}
\end{equation}
With this notation, we get
\begin{align}
\Var_q(\overline{Y}-\mu\overline{Z})
& =  \frac{1}{N^2} \left ( 
\sum_{i = 1}^N \Var_q\Bigl(g(X_i)\Bigr)  + 2 \sum_{i < j } \Cov_q \Bigl (g(X_i), g(X_j) \Bigr) \right ). 
\end{align}
If the variables $X_i$ form a stationary stochastic process, as in a converged MCMC algorithm, then $\Var_q \Bigl(g(X_i)\Bigr)$ and  $\Cov_q\Bigl(g(X_i), g(X_{i+k})\Bigr)$ depend only on $k$ and not $i$.
Hence, with $X\coloneqq X_1$,
\begin{align}
\Var_q(\overline{Y}-\mu\overline{Z})
& = 
\Var_q\Bigl(g(X)\Bigr) \left ( \frac{1  + 2 \sum_{k=1}^{N-k} \rho_{g}(k)}{N} \right ), \label{Eq:Var_g_n-k}
\end{align}
where $\rho_{g}$ is the autocorrelation function of the stationary process.
According to \citep{Geyer_1992}, for large $N$, 
the variance is approximately 
\begin{align}
\Var_q(\overline{Y}-\mu\overline{Z})
& \approx
\Var_q\Bigl(g(X)\Bigr) \left ( \frac{1  + 2 \sum_{k=1}^\infty \rho_{g}(k)}{N} \right ). \label{Eq:Var_g}
\end{align}
If $\sum_{k=1}^\infty \rho_{g}(k)$ converges, then a standard approximation is \citep{Geyer_1992, Givens_2012}
\begin{equation}
\sum_{k=1}^\infty \rho_{g}(k)\approx \sum_{k=1}^\ell \hat{\rho}_{g}(k), 
\end{equation}
where $\ell$ is the first index for which $\hat{\rho}_{g}(\ell+1)<0$,
and where $\hat{\rho}_g$ is the empirical autocorrelation function of
the sample $g(X_1)$, \dots, $g(X_N)$. 
Note that evaluating $g$ requires knowledge of $\mu$ which is precisely
the quantity we wish to estimate. So, instead we will use
\begin{equation}
    \tilde{g}(X_i)
    \coloneqq (f(X_i) - \widetilde{\mu})w_p(X_i), 
\end{equation}
as an approximation for $g(X_i)$, and therefore use
\begin{equation}
\label{eq:approxsumrho}
\sum_{k=1}^\infty \rho_{g}(k)\approx \sum_{k=1}^\ell \hat{\rho}_{\tilde{g}}(k).
\end{equation}

We now estimate $\Var_q\Bigl(g(X)\Bigr)$ in \cref{Eq:Var_g}. Using \cref{Eq:g_tilde}, we also get that
\begin{align}
\Var_q\Bigl(g(X)\Bigr) &= \Var_q\Bigl(f(X) w_p(X) \Bigr) - 2 \mu \Cov_q \Bigl(f(X) w_p(X) , w_p(X)\Bigr) \nonumber \\ 
&\qquad + \mu^2 \Var_q \Bigl(w_p(X)\Bigr). 
\label{Eq:Var_g_tilde_back}
\intertext{Now, following the same ideas as in \citep{Elvira_2018}[p.~5-6] (see \ref{Calculations} for details) we get }
\Var_q\Bigl(g(X)\Bigr) & \approx N \Var_p(\overline{\mu})\left(\frac{1}{N}\sum_{i=1}^N w_p^2(X_i)\right).
\end{align}
Putting everything together, we get
\begin{align}
 \Var_q(\widetilde{\mu})  &\approx  N \Var_p(\overline{\mu})  \frac{\frac{1}{N}\sum_{i=1}^N w_p^2(X_i)}{\left(\frac{1}{N}\sum_{i=1}^N w_p(X_i)\right)^2}
 \left ( \frac{1  + 2\sum_{k=1}^\ell \hat{\rho}_{\tilde{g}}(k)}{N} \right ), \\ 
 & =  \frac{N^2\Var_p(\overline{\mu})}{\mathrm{ESS_{IS}} \cdot \mathrm{ESS_{MCMC}}},
 \label{Eq:Final_Var_Importance_sampling}
\end{align}
where $\mathrm{ESS_{IS}}$ is the standard estimate of the effective sample size for importance sampling with independent samples, \cref{Eq:ESS_aprox},
and $\mathrm{ESS_{MCMC}}$ is the MCMC effective sample size (for $\tilde{g}$),
\begin{equation}
  \mathrm{ESS_{MCMC}}\coloneqq
  \frac{N}{1  + 2\sum_{k=1}^\ell \hat{\rho}_{\tilde{g}}(k)}. \nonumber
\end{equation}

Substituting \cref{Eq:Final_Var_Importance_sampling} into \cref{Eq:ESS},
we finally obtain the following estimate for the combined $\mathrm{ESS}$
\begin{align}
\mathrm{ESS}
& \approx 
\frac{ \mathrm{ESS_{MCMC}}}{N} \cdot \mathrm{ESS_{IS}}.
\label{Eq: N_eff_estimated_MCMC_IS}
\end{align}
So, what is new in \cref{Eq: N_eff_estimated_MCMC_IS} with respect to \cref{Eq:ESS_aprox} is 
the factor $\frac{\mathrm{ESS_{MCMC}}}{N}$ which accounts for
a reduction in effective sample size
due to the correlation of the MCMC samples (a reduction since it is very 
unlikely that the correlation will be negative).

\section{Importance sampling over a set of probability distributions}
\label{sec:IS_sets}

Let $\mathcal{M} \subset \mathbb{R}^d$ be a compact set and $p_t(\cdot) = \{p(\cdot|t)| t \in  \mathcal{M}$\} be a probability density function parameterized  by $t$ (i.e. hyperparameters).
The lower expectation of a function $f$ with respect to $p_t$ for all $t \in  \mathcal{M}$ 
is assumed to exist and is estimated by
\begin{equation}
\underline{\mathrm{E}}(f) \coloneqq \displaystyle \min_{t \in \mathcal{M}} \int f(x)p_t(x) dx. \label{Lower_Expectation}
\end{equation}

In practice, one can search for the minimum
using numerical methods. For instance, an iterative version of standard importance sampling has been used to estimate bounds on expectations in robust Bayesian analysis \citep{Fetz_2017,Troffaes_2017,Troffaes_2018,Fetz_2018}.

Using importance sampling, the lower expectation is estimated by 
\begin{equation}
\underline{\widehat{\mathrm{E}}}(f)  \approx \displaystyle \min_{t \in \mathcal{M}} \frac{\sum_{i=1}^N f(X_i)w_t(X_i)}{\sum_{i=1}^N w_t(X_i)}, \label{Lower_Expectation_IS}
\end{equation}
where $w_t(x) \coloneqq \frac{c p_t(x)}{q(x)}$.

Iterative importance sampling estimates the lower expectation 
of a function $f$ by moving the sampling distribution, $q$, towards the optimal distribution 
\citep{Fetz_2017,Troffaes_2017}. The stopping criterion is that the effective sample size of importance sampling 
should be close enough to the desired independent sample size (denoted by $\mathrm{ESS_{target}}$) that is fixed in advance.
In this paper, we adapt iterative importance sampling introduced in \citep{Fetz_2017,Troffaes_2017,Troffaes_2018}. 
Our contribution is the effective sample size of importance sampling with correlated MCMC samples which allows us to combine iterative importance sampling with MCMC sampling, thus allowing for robust analysis of more complex models. 
It is important to highlight that how large the sample size for MCMC samples should be and $\mathrm{ESS_{target}}$ in step 2 and step 4 respectively, are values fixed in advance (i.e. they are inputs to the procedure).
The method goes as follows:

\begin{description}
\item[Step 1] Set $t = t^0$ where $t^0$ is an initial value in the feasible region.
\item[Step 2] Generate samples from $q(x) = p_t(x)$ using MCMC sampling until the effective sample size for the MCMC sample is large enough (i.e. exceed a specified threshold). 
\item[Step 3] Find $t_* = \arg \min_{t \in \mathcal{M}} \frac{\sum_{i=1}^N f(X_i)w_t(X_i)}{\sum_{i=1}^N w_t(X_i)}$ using an optimization algorithm.
\item[Step 4] If $\mathrm{ESS} > \mathrm{ESS_{target}}$, or maximum number of iterations reached, then stop.
\item[Step 5] Set $t = t_*$ and go to Step 2.
\end{description}

The convergence of the method depends on the distributions and the parameter space. We have set a maximum number of iterations (i.e. in our case, 10 000) to stop the algorithm when it does not converge.

The condition of a large enough sample size is added in Step 2 to ensure that the 
optimization in Step 3 is based on a reliable sample. For example, in the application below, we first specify the $\mathrm{ESS_{target}}$ and then we require the effective sample size for the MCMC samples to be 20\% greater than the $\mathrm{ESS_{target}}$.
The stopping criterion in Step 4 uses the effective sample size for importance sampling with correlated samples \cref{Eq: N_eff_estimated_MCMC_IS}, 
controlling for both convergence of the MCMC and quality of importance sampling.
The reason for this stopping criterion is that MCMC sampling might require different numbers of iteration to produce a reasonable number of efficient samples.

Thinning chains in MCMC has been used in several papers; see for instance \citep{Croll_2006, Endo_2019, Kopylev_2009}. Thinning consists of taking every k-th sample instead of all of them in order to reduce autocorrelation. In \citep{Link_2012} it is shown that although thinning chains in MCMC reduces autocorrelation between MCMC samples, it also reduces the precision of the estimates (i.e.  the average over a thinned sample set has greater variance than the average over the unthinned sample) \citep{Geyer_1992}. Therefore, thinning is not advisable unless it is needed due to computer memory limitations \citep{Link_2012}. 

\section{An application}

To illustrate the proposed method for robust Bayesian analysis, iterative importance sampling with MCMC sampling is applied on a previously published meta-analysis investigating the effect of biomanipulation (the intervention) on water quality in freshwater lakes \citep{Bernes_2015}. 
We selected the random effects model (described below) for the meta-analysis of the change in the level of \emph{Chlorophyll a} before and during biomanipulation \citep{Bernes_2015}. Available data are estimated mean differences and estimation errors from 75 studies. The estimated effects range from -24.17 to 332.50 $\mu g/ l$ with a sample mean of 28.46 $\mu g/ l$. 
The 5th and 95th percentile of the data are -11.10 and 76.29 $\mu g/ l$ respectively. Here, a positive value corresponds to an improvement in water quality by biomanipulation (since we have turned the sign of the data). 
In addition, we investigate what happens with the performance of the suggested method when 
the set of priors is changed from a set with low to high prior-data conflict.

\subsection{A Bayesian Linear Random Effects Model} \label{subsection Linear RE}

The overall effect of the intervention $\mu$ is estimated by a linear random effects model (\Cref{fig: Bayesian RE meta-analysis}) according to
\begin{align}
y_i | \delta_i & \sim N(\delta_i, \sigma_i^2),   \\
\delta_i | \mu, k, \tau_{\mu} & \sim N(\mu, k^2 \hspace{0.5mm} \tau_{\mu}^2), \label{marginal_dist}
\end{align}
where $y_i$ is the observed intervention effect in study $i$ ($i = {1, \dots, N})$ with known within-study variance $\sigma_i^2$,
$\delta_i$ is the specific intervention effect in study $i$, and $k^2 \hspace{0.5mm} \tau_{\mu}^2$ is the between-study variance.
This model can be expressed in its marginal form \citep{Gelman_2013, Rover_2017} as 
\begin{align}
y_i | \mu, k, \tau_{\mu} \sim N(\mu, \sigma_i^2 + k^2 \hspace{0.5mm} \tau_{\mu}^2). 
\end{align}
To implement this model in a Bayesian framework, we select 
the following prior distributions for the parameters  
\begin{align}
\mu | \tau_{\mu} & \sim N(\mu_0,\tau_{\mu}^2), \label{Eq:prior1}  \\
\tau_{\mu} & \sim U(\tau_l,\tau_0), \label{Eq:prior2} \\    
k & \sim  U(k_l, k_u),\label{Eq:prior3} 
\end{align}
where $\mu_0$ is the mean of the overall intervention effect, $\tau_{\mu}$ is the 
standard deviation of the overall intervention effect which ranges between $\tau_l$ 
and $\tau_0$, 
$k$ is a proportionality constant which ranges between $k_l$ 
and $k_u$. 
We let $\tau_l = 1$, $k_l = 1$ and $k_u = 5$.

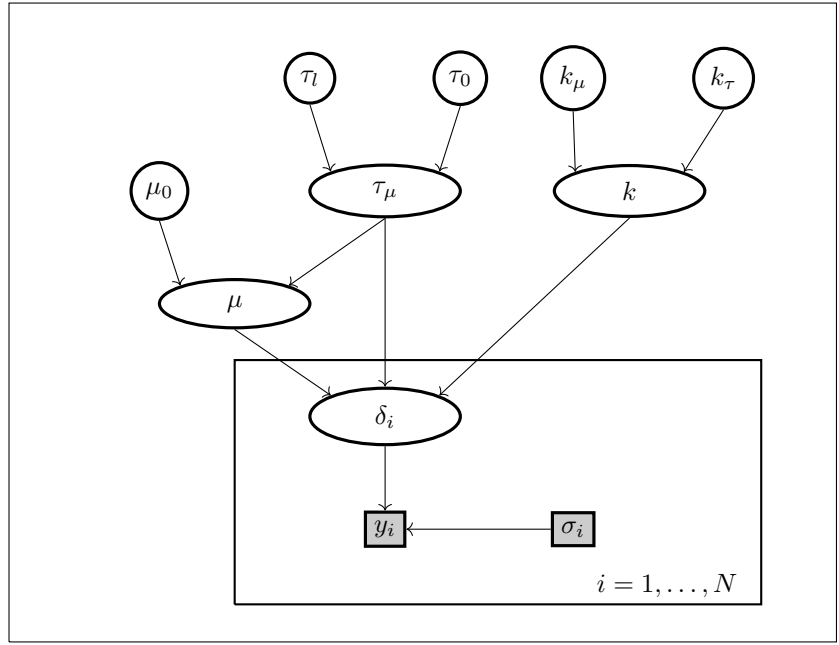
\begin{figure}[ht]
\begin{center}
			\begin{tikzpicture}[
	   			 ellipsenode/.style={ellipse, minimum width=2cm, minimum height=0.5cm, draw=black, very thick},
				 roundnode/.style={circle, draw=black, very thick, minimum size=5mm}, 
				 squarednode/.style={rectangle, draw=black, very thick, minimum size=3mm},
				 squarednode1/.style={rectangle, draw=black, fill=black!20, very thick, minimum size=3mm},
				 ]	
				 Nodes
				 \node[roundnode, align=center]          		(tau_low_level0)        	at (6, 21)   		{$\tau_l$};
				 \node[roundnode, align=center]          		(tau_0_level0)         		at (8, 21)   		{$\tau_0$};
				 \node[roundnode, align=center]          		(k_mu_level0)         		at (9.5, 21)   		{$k_\mu$};
				 \node[roundnode, align=center]          		(k_sd_level0)         		at (11.5, 21)   	{$k_{\tau}$};
				 \node[roundnode, align=center]          		(mu_mu_level1)         		at (4, 19.5)   		{$\mu_0$};
				 \node[ellipsenode, align=center]          		(mu__tau_level1)         	at (7, 19.5)   		{$\tau_{\mu}$};
				 \node[ellipsenode, align=center]         		(k_level1)      			at (10.25, 19.5)   	{$k$};
				 \node[ellipsenode, align=center]          		(mu_level2)        			at (5, 18)			{$\mu$};
				 \node[ellipsenode, align=center]          		(delta_level3)       		at (7, 16.5) 		{$\delta_i$};
				 \node[squarednode1, align=center]        	 	(y_level4)       			at (7, 15) 			{$y_i$};
				 \node[squarednode1, align=center]          	(sigma_level4)       		at (9.5, 15) 		{$\sigma_i$};
				 \node[align = center]							(indice_level6)  			at (10.75, 14.25)   {$i = 1, \dots, N$};
				 	 
				 Lines
				 \draw[->] (mu_mu_level1.south) 		-- 	(mu_level2.north west);
				 \draw[->] (mu__tau_level1.south) 		-- 	(mu_level2.north east);
				 \draw[->] (mu__tau_level1.south) 		-- 	(delta_level3.north);
				 \draw[->] (k_level1.south) 			-- 	(delta_level3.north east);
				 \draw[->] (k_mu_level0.south)		 	-- 	(k_level1.north west);
				 \draw[->] (k_sd_level0.south) 			-- 	(k_level1.north east);
				 \draw[->] (tau_low_level0.south)		-- 	(mu__tau_level1.north west);
				 \draw[->] (tau_0_level0.south) 		-- 	(mu__tau_level1.north east);
				 \draw[->] (mu_level2.south) 			-- 	(delta_level3.north west);
				 \draw[->] (delta_level3.south)			-- 	(y_level4.north);
				 \draw[->] (sigma_level4.west) 			-- 	(y_level4.east);
				 \draw[thick,-]  (5, 14) -- (5, 17.25) -- (12, 17.25) -- (12, 14) -- (5, 14); 
				 \draw[-]  (2, 13.5) -- (2, 22) -- (13, 22) -- (13, 13.5) -- (2, 13.5);
			\end{tikzpicture}
   	\end{center}
   	\caption{ Linear random effects graphical model. Unknown quantities (parameters) are represented by ellipses, known quantities (priors) by circles and observed data by filled squares. The plate indicates repeated cases.}
   	\label{fig: Bayesian RE meta-analysis}
	\end{figure}
	
The linear random effects model is implemented using MCMC sampling in Stan through 
the \emph {rstan} package, the R interface to Stan \citep{RStan} 
(see supplementary material for Stan code).

\subsection{Selecting a set of priors}

To expand the linear random effects model into a robust Bayesian framework, we consider a set of prior distributions
for $\mu$ and $\tau_\mu$. 

There are several methods to elicit priors from experts \citep{OHagan_2006, Hanea_2021, Hartmann_2020, SHELF_2018}, but there is no obvious alternative to specify a set of priors. The approach that we use to specify a set of priors for multiple parameters is chosen to illustrate robust Bayesian analysis with IIS and MCMC sampling, and other approaches could have been used as well. 
The prior distributions for each parameter are assumed to belong to the same family of probability
distributions (eq. \ref{Eq:prior1} and eq. \ref{Eq:prior2}), but with different hyperparameters. 
In order to consider interaction between parameters, the specification of priors is made using an approach similar to prior predictive check \citep{Daimon_2008, Schad_2019}.  
A compact set of hyperparameters is selected by comparing the distribution of the intervention effect for a random study conditional on the hyperparameters $\delta^*|\mu,\tau_{\mu}$ to an elicited range $R$. 

The selection of a set of priors can be summarized as:  
\begin{enumerate}
\item Specify a range $R$ where the effect size of a randomly selected study is expected to fall with a probability of at least h\% (i.e. target coverage).
\item Specify a regular grid of hyperparameters $(\mu_0, \tau_0)$. 
\item For each combination of hyperparameters $(\mu_0, \tau_0)$, generate a random sample $\delta^*_1, \dots, \delta^*_M$ from \cref{marginal_dist}.
\item Identify hyperparameters where the proportion of generated samples falling inside $R$ 
exceeds the target coverage $h$.
\item Find a function that discriminates hyperparameters complying with the target coverage and use the function to select a compact set of hyperparameters.
\end{enumerate}

Here, we select two sets of priors that represent situations with a low and high prior-data conflict, respectively. In actual applications, priors should be elicited by structured expert judgement \citep{OHagan_2006}.

A set of priors that represents a situation with low prior-data conflict
is derived using an elicited range of specific intervention effect of $R = [-20,80]$, a target coverage of $90\%$, and a regular grid of $200 \times 200$ of hyperparameters ($-100 \leq \mu_0 \leq 100$ and $5 \leq \tau_0 \leq 50)$. 
Following the procedure described above yields to the set 
\begin{equation}
\mathcal{M} =  \left\{\begin{array}{c}
										-8 \leq \mu_0 \leq 68 \\ 
										5 \leq \tau_0 \leq 16 \\
										r(\mu_0,\tau_0) - 0.9 \geq 0
										\end{array} \right \}, \nonumber
\end{equation}
where $r(\mu_0,\tau_0) = \tau_0 - (-0.01 \mu_0^2 + 0.46 \mu_0 + 9.56)$ is a function used to capture hyperparameters fulfilling the target coverage (\Cref{fig: Prior_Check}).

\begin{figure}
\centering
\resizebox*{9cm}{!}{\includegraphics{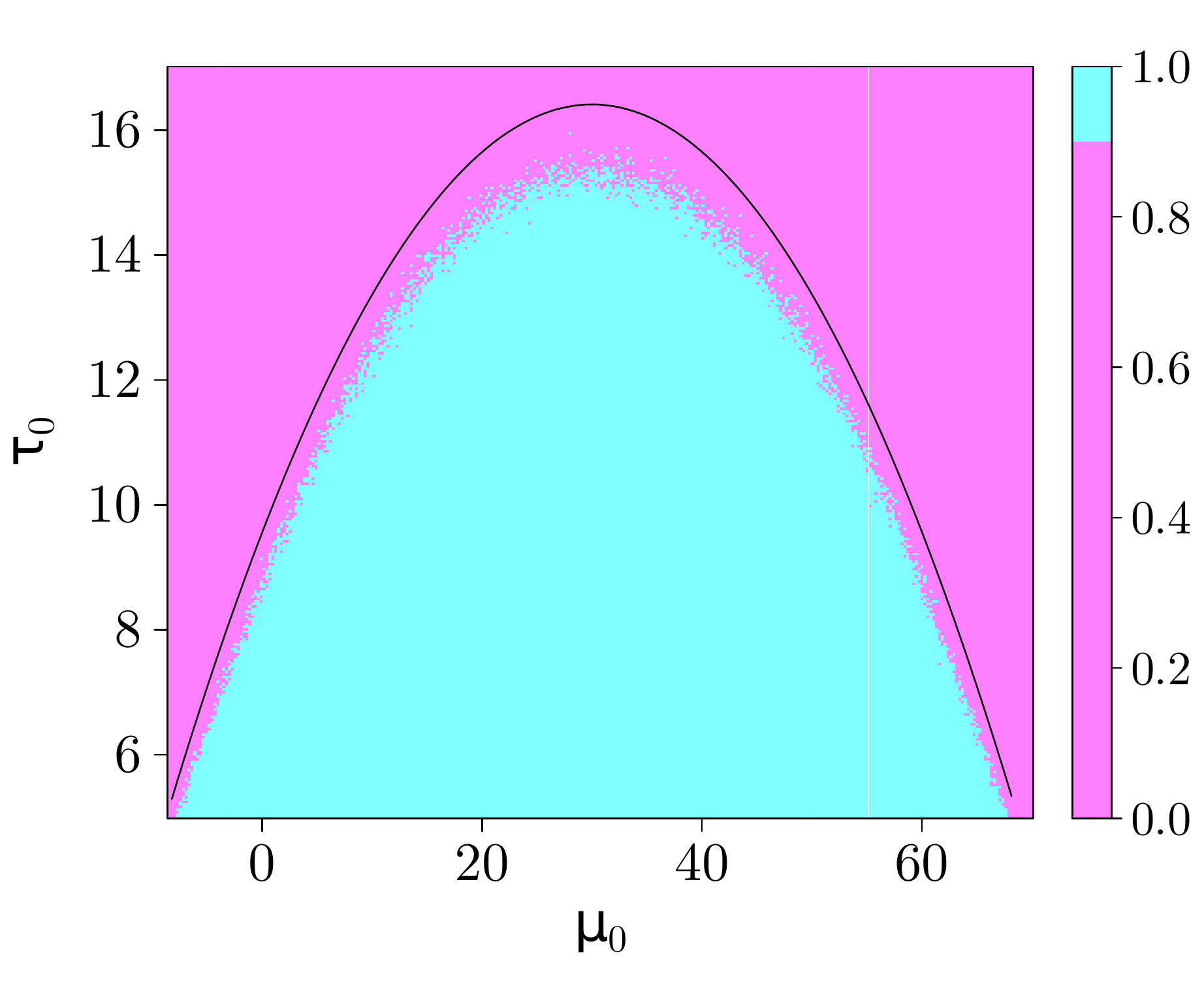}}
\caption{Prior check. Selection of priors, a curve fitted to the area where more than $90\%$ of the generated specific intervention effect fell inside of the elicited range.} \label{fig: Prior_Check}
\end{figure}

\subsection{Setting up the iterative importance sampling}

The quantity of interest is the expected value of the overall intervention effect ($\mu$) on the linear random effects model given in \cref{subsection Linear RE}. Iterative importance sampling is applied to estimate the lower bound of the quantity of interest over the compact 
domain of priors $\mathcal{M}$ (see \ref{proof_minimum} for the proof of the existence of the minimum in our example). The lower bound of $\mu$ is 
\begin{align}
		\underline{\widehat{\mathrm{E}}}(\mu|\mathbf{y}, \mu_0, \tau_0) 
		& \approx \min_{(\mu_0, \tau_0) \in \mathcal{M}} \frac{\sum_{i = 1}^N \mu^{(i)} w(X^{(i)}; \mu_0, \tau_0)}{\sum_{i = 1}^N w(X^{(i)}; \mu_0, \tau_0)},
\end{align}
where $X^{(i)} = \left \{ \mu^{(i)}, k^{(i)}, \tau_\mu^{(i)} \right \}$ and $\mu^{(i)}$ are 
samples drawn from a sampling distribution using MCMC based on hyperparameters $\mu_0^\prime$ 
and $\tau_0^\prime$. The weights are given by
\begin{align}
w(X^{(i)}; \mu_0, \tau_0) & = \frac{p(\mu^{(i)}, k^{(i)}, \tau_\mu^{(i)}| \mathbf{y}, \mu_0, \tau_0)}{p(\mu^{(i)}, k^{(i)}, \tau_\mu^{(i)}| \mathbf{y}, \mu_0^\prime, \tau_0^\prime)},
\label{expw}
\end{align}
where $p(\mu, k, \tau_\mu| \mathbf{y}, \mu_0, \tau_0)$ is the posterior distribution of the 
target distribution corresponding to hyperparameters $\mu_0$ and $\tau_0$. 

Using \cref{expw} as weights in \cref{Eq: N_eff_estimated_MCMC_IS} gives the effective 
sample size of importance sampling with MCMC samples.  
The search for the lower bound in Step 3 will stop when the optimal prior is close 
to the prior used in the MCMC (\Cref{fig: N_eff_contourplot}). 
How close is determined by the assigned $\mathrm{ESS_{target}}$ in Step 4, which 
in our example is set to 5 000 and 10 000 samples.
The effective sample size for the MCMC samples  
in Step 2 is set to exceed $\mathrm{ESS_{target}}$ by at least 20\%, i.e. 
MCMC sampling in Step 2 will be run until an effective sample size of 
6 000 and 12 000 samples has been reached.

\begin{figure}
\centering
\resizebox*{4.7cm}{!}{\includegraphics{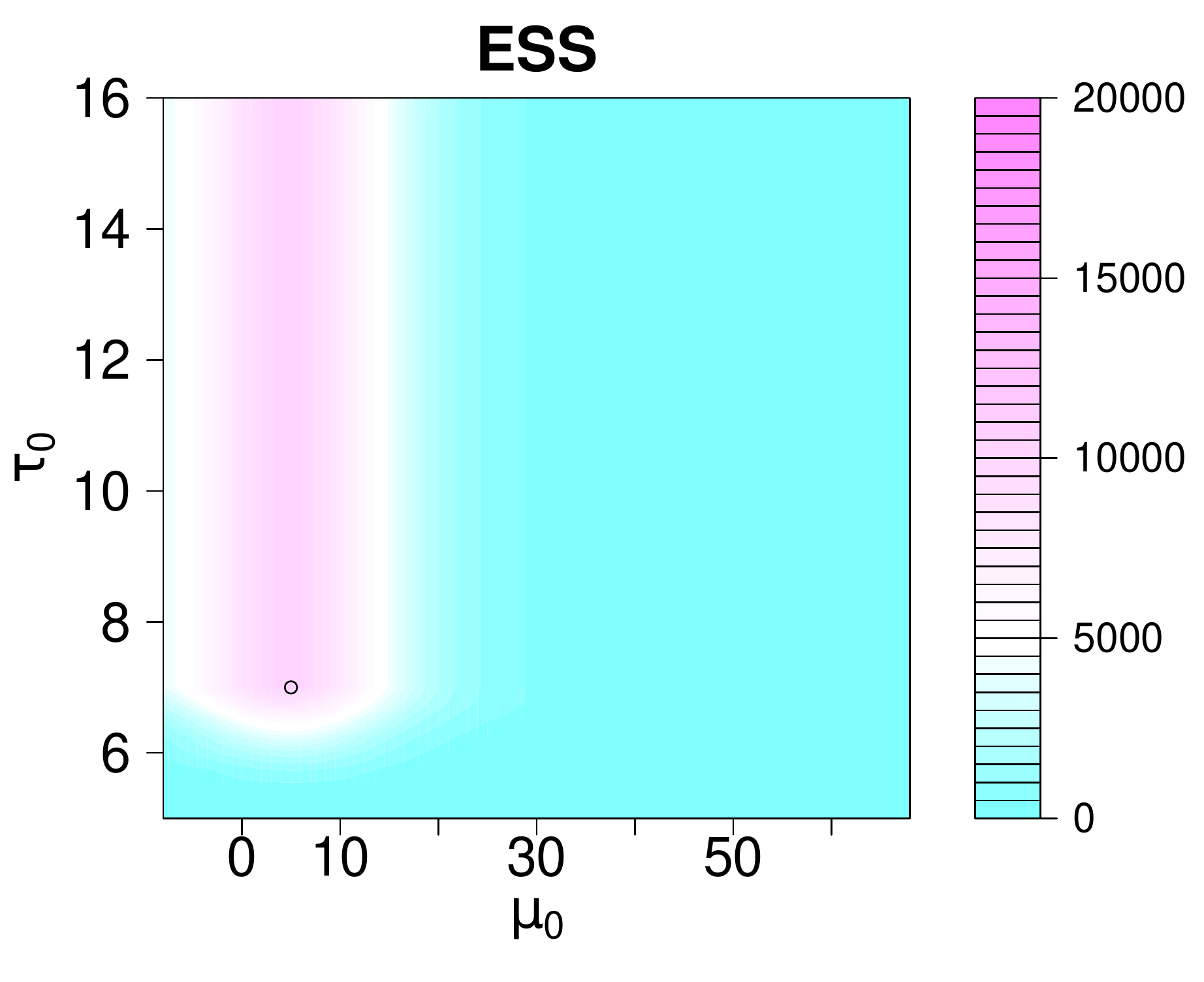}}
\resizebox*{4.7cm}{!}{\includegraphics{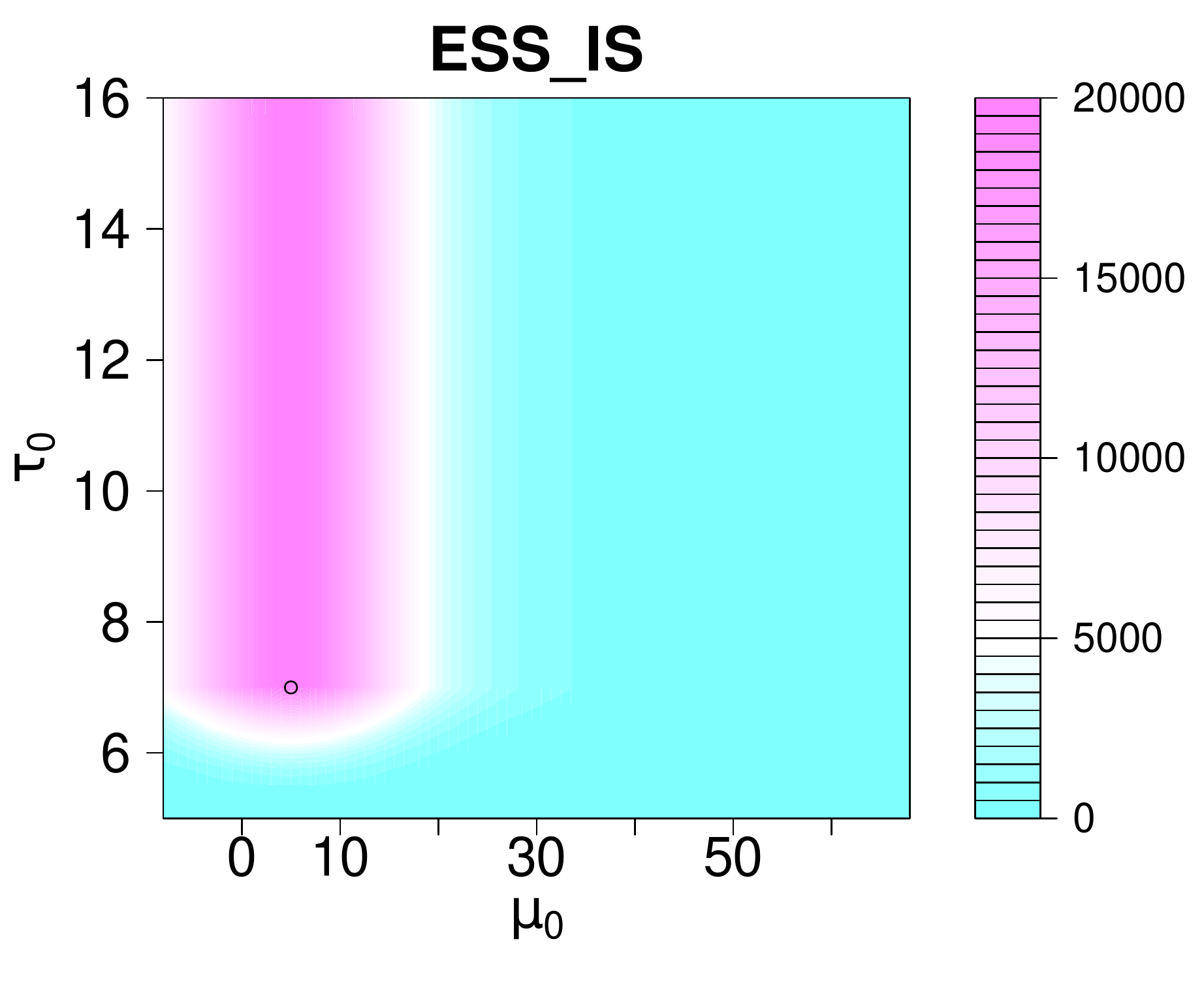}}
\resizebox*{4.7cm}{!}{\includegraphics{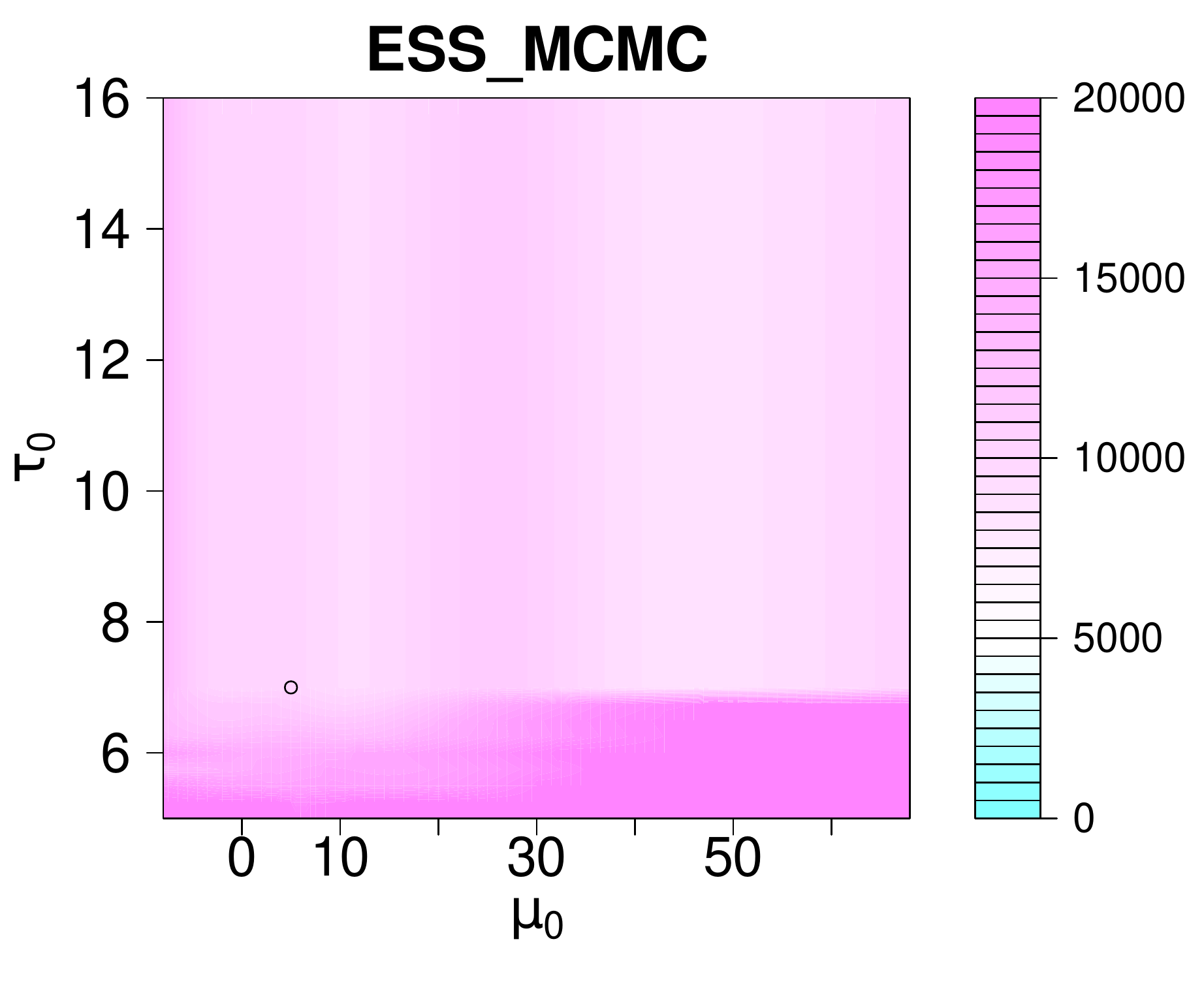}}
\caption{From left, effective sample sizes $\mathrm{ESS}$, $\mathrm{ESS_{IS}}$ and 
	$\mathrm{ESS_{MCMC}}$ given MCMC samples using hyperparameters $\mu_0 = 5$ and $\tau_0 = 7$ 
	(the dot). Here, $\mathrm{ESS_{target}}$ is 5 000 and if the optimization in Step 3 moves the 
	hyperparameters into a region with a lower $\mathrm{ESS}$ we need to re-run the MCMC in Step 2,
	 using the `new' hyperparameters.} \label{fig: N_eff_contourplot}
\end{figure}

The optimization in Step 3 is performed by simulated annealing, \citep{Du_2016, Henderson_2003}, 
using the \emph{optim} function from the \emph {optimx} package \citep{Nash_2011} in R.

\subsection{The lower bound of the expected overall effect}

We illustrate the method using two different $\mathrm{ESS_{target}}$ and two choices of initial values $(\mu_0^\prime = -7$, $\tau_0^\prime = 6)$ and $(\mu_0^\prime = 10$, $\tau_0^\prime = 10)$. The number of iterations of iterative importance sampling ranges between 2 and 3 with an effective sample size between $11~640$ and $14~305$, see (\Cref{Tab: All data}, \ref{Tab: All data1} and \ref{Tab: All data2}). Note that the effective samples size for the MCMC sample is derived based on the hyperparameters in the previous iteration.

The estimated lower bound on the expected 
overall effect over $\mathcal{M}$ is 12.1 $\mu g/ l$ (\Cref{Tab: All data}, \ref{Tab: All data1} and  \ref{Tab: All data2}). 
We also run the method using different random seeds which gives an estimated lower bound on the expected overall effect between 12.07 and 12.24 $\mu g/ l$.
The bound is lower than 19.9 
which is the estimated effect in a standard Bayesian analysis using a flat prior centered at zero ($\mu_0 = 0$ and $\tau_0 = 1~000$). The lower bound is obtained at the corner of the prior region (i.e. $\mu_0 = -8$ and  $\tau_0 = 5$).
The bound obtained by iterative importance sampling is compared to one estimated by a grid search method across a regular grid of hyperparameters (which yields a total of 4 773 hyperparameter values). The grid search method gives a lower expected bound of 12.09 $\mu g/ l$ which is close to the one found using iterative importance sampling. It takes much longer time to run a grid search method, over 10 hours, compared to the iterative importance sampling taking roughly 15 to 25 minutes (\Cref{Tab: All data}, \ref{Tab: All data1} and  \ref{Tab: All data2}). 

Moreover, we run the optimization algorithm (simulated annealing) over the set $\mathcal{M}$, but without the resampling step, (e.g. just re-run the MCMC for each parameter that is evaluated). The method gives a lower bound of 12.13 $\mu g/ l$ and it takes much longer time to run, over 20 hours.

The iterative importance sampling converges with fewer iterations when the initial values are relative close to the hyperparameters corresponding to the lower bound. 
\begin{table}[ht]
\caption{Summary of iterative importance sampling (IIS) of the lower expected value over the set $\mathcal{M}$, $\mathrm{ESS_{target}}$ = 5 000 and initial values $(\mu_0; \tau_0) = (-7; 6)$.} 
\label{Tab: All data} \par
\vskip .2cm
\centerline{\tabcolsep=3truept
\begin{tabular}{c  c c c c c c c c c}
\hline
\textbf{} & Iter & Samples & $\mu_{0_*}$ & $\tau_{0_*}$  & $\mathrm{ESS}$ & $\mathrm{ESS_{MCMC}}$ & $\mathrm{ESS_{IS}}$ & $\hat{\mu}_{(\mu_{0_*}, \tau_{0_*})}$ & Time \\ 
\hline
\textbf{Grid} & \multirow{2}{*}{--} & \multirow{2}{*}{20 000} & \multirow{2}{*}{-8} & \multirow{2}{*}{5} & \multirow{2}{*}{--} & \multirow{2}{*}{--} & \multirow{2}{*}{--} & \multirow{2}{*}{12.091} & \multirow{2}{*}{11.45 h} \\
\textbf{Search} & & & & & &  & & & \\
\hline
\multirow{3}{*}{\textbf{IIS}} & 0 & -- & -7 & 6 & -- & -- & -- & --  & -- \\
& 1 & 20 000 & -7.706 & 5.051 & 27 & 18 005 & 30 & 10.988  & -- \\
& 2 & 20 000 & -7.999 & 5.045 & 11 640 & 12 307  & 18 916 & 12.166  & 11.78 mins \\
\hline
\end{tabular}}
\end{table}

\begin{table}[ht]
\caption{Summary of iterative importance sampling (IIS) of the lower expected value over the set $\mathcal{M}$ and $\mathrm{ESS_{target}}$ = 10 000 and initial values $(\mu_0; \tau_0) = (-7; 6)$.}
\label{Tab: All data1} \par
\vskip .2cm
\centerline{\tabcolsep=3truept
\begin{tabular}{c  c c c c c c c c c}
\hline
\textbf{} & Iter & Samples & $\mu_{0_*}$ & $\tau_{0_*}$  & $\mathrm{ESS}$ & $\mathrm{ESS_{MCMC}}$ & $\mathrm{ESS_{IS}}$ & $\hat{\mu}_{(\mu_{0_*}, \tau_{0_*})}$  & Time \\ 
\hline
\textbf{Grid} & \multirow{2}{*}{--} & \multirow{2}{*}{20 000} & \multirow{2}{*}{-8} & \multirow{2}{*}{5} & \multirow{2}{*}{--} & \multirow{2}{*}{--} & \multirow{2}{*}{--} & \multirow{2}{*}{12.091} & \multirow{2}{*}{11.45 h} \\
\textbf{Search} & & & & & &  & & & \\
\hline
\multirow{3}{*}{\textbf{IIS}} & 0 & -- & -7 & 6 & -- & -- & -- & -- & -- \\
& 1 & 40 000 & -7.900 & 5.035 & 59 & 31 579 & 75 & 12.006  & -- \\
& 2 & 20 000 & -7.998  & 5.136 & 13 899  & 13 901 & 19 997 & 12.135 &  18.73 mins\\
\hline
\end{tabular}}
\end{table}

\begin{table}[ht]
\caption{Summary of iterative importance sampling (IIS) of the lower expected value over the set $\mathcal{M}$ and $\mathrm{ESS_{target}}$ = 10 000 and initial values $(\mu_0; \tau_0) = (10; 10)$.}
\label{Tab: All data2} \par
\vskip .2cm
\centerline{\tabcolsep=3truept
\begin{tabular}{c  c c c c c c c c c}
\hline
\textbf{} & Iter & Samples & $\mu_{0_*}$ & $\tau_{0_*}$  & $\mathrm{ESS}$ & $\mathrm{ESS_{MCMC}}$ & $\mathrm{ESS_{IS}}$ & $\hat{\mu}_{(\mu_{0_*}, \tau_{0_*})}$ & Time \\ 
\hline
\textbf{Grid} & \multirow{2}{*}{--} & \multirow{2}{*}{20 000} & \multirow{2}{*}{-8} & \multirow{2}{*}{5} & \multirow{2}{*}{--} & \multirow{2}{*}{--} & \multirow{2}{*}{--} & \multirow{2}{*}{12.091} & \multirow{2}{*}{11.45 h} \\
\textbf{Search} & & & & & &  & & & \\
\hline
\multirow{3}{*}{\textbf{IIS}} & 0 & -- & 10 & 10 & -- & -- & -- & -- & --\\
& 1 & 40 000 & -6.734 & 5.148 & 5 & 40 000 & 5 & 12.762 & -- \\
& 2 & 20 000 & -7.993 & 5.008 & 3 995 & 13 907 & 5 746 & 12.145 & -- \\
& 3 & 20 000 & -7.995 & 5.221 & 14 305 & 14 305  & 20 000 & 12.103 & 23.72 mins \\
\hline
\end{tabular}}
\end{table}

\subsection{The influence of prior data conflict}

In order to illustrate what happens when the set of prior has a high conflict with data, we specify a new set $\mathcal{M'}$ using the procedure described in subsection 5.2, but with $R=[30, 100]$. This range does not include the median (14.6 $\mu g/ l$) of the observed effects in the 75 studies or the posterior mean from a standard Bayesian analysis with a flat prior. This gives the set 
\begin{equation}
\mathcal{M'} =  \left\{\begin{array}{c}
										 42 \leq \mu_0 \leq 88 \\ 
										  5 \leq \tau_0 \leq  11 \\
										r'(\mu_0,\tau_0) - 0.90 \geq 0
										\end{array} \right \}, \nonumber
\end{equation}
where $r'(\mu_0,\tau_0) = \tau_0 - (-0.011 \mu^2_0 + 1.427 \mu_0 - 35.104)$.

Shifting the prior region from $\mathcal{M}$ to $\mathcal{M'}$, (i.e. towards higher values of $\mu$), increases the lower 
bound of the expected overall effect from 12.1 to 26.8 $\mu g/ l$, (\Cref{Tab: elicited range of data_1}). The lower bound is in this case obtained for a less precise prior $\tau_0 = 10.5$.  

We also compare the bound obtained by iterative importance sampling to the one estimated by a grid search method across a regular grid of hyperparameters (which yields a total of 1 649 hyperparameter values). The grid search method gives a slightly lower expected bound of 26.7 $\mu g/ l$, but takes much longer time, 4 hours, compared to 13 minutes for the iterative importance sampling (\Cref{Tab: elicited range of data_1}).

\begin{table}[ht]
\caption{Summary of iterative importance sampling of the lower expected values over the set $\mathcal{M'}$ and $\mathrm{ESS_{target}}$ = 10 000 and initial values $(\mu_0; \tau_0) = (60; 11)$.}
\label{Tab: elicited range of data_1} \par
\vskip .2cm
\centerline{\tabcolsep=3truept
\begin{tabular}{c  c c c c c c c c c}
\hline
\textbf{} & Iter & Samples & $\mu_{0_*}$ & $\tau_{0_*}$  & $\mathrm{ESS}$ & $\mathrm{ESS_{MCMC}}$ & $\mathrm{ESS_{IS}}$ & $\hat{\mu}_{(\mu_{0_*}, \tau_{0_*})}$ & Time \\ 
\hline
\textbf{Grid} & \multirow{2}{*}{--} & \multirow{2}{*}{20 000} & \multirow{2}{*}{57} & \multirow{2}{*}{10.5} & \multirow{2}{*}{--} & \multirow{2}{*}{--} & \multirow{2}{*}{--} & \multirow{2}{*}{26.754} & \multirow{2}{*}{3.70 h} \\
\textbf{Search} & & & & & &  & & & \\
\hline
\multirow{2}{*}{\textbf{IIS}} & 0 & -- & 60 & 11 & -- & -- & -- & -- & -- \\
& 1 & 40 000 & 59.064  & 10.859 & 15 371 & 17 533 & 35 067 &  26.811  & 12.5 mins \\
\hline
\end{tabular}}
\end{table}

\section{Conclusions}
\label{sec:conclusions}

Robust Bayesian analysis (i.e. Bayesian inference under a set of priors) offers a way to quantify epistemic uncertainty by bounded instead of precise probability. This type of analysis is useful for evaluating sensitivity to the choice of prior. It is also a solution for Bayesian inference when prior information is upfront given as a set of distributions.

We have derived an expression for the effective sample size of importance sampling with correlated MCMC samples, which ensures reliable samples from both MCMC and importance sampling procedures when combining them. The combination of iterative importance sampling with MCMC sampling was used for robust Bayesian analysis on an existing meta-analysis based on a random effects model \citep{Bernes_2015}. The estimated lower bound on the expected overall effect of the intervention in the meta-analysis is a conservative estimate compared to what would have been the result from selecting one prior in the set and using a standard Bayesian analysis.

Iterative importance sampling with MCMC sampling 
allows for robust Bayesian analysis on a wider range of models not limited to conjugate models. The flexibility in the choice of model may allow for more applications of robust Bayesian analysis. The method was demonstrated on a relatively simple model with two different choices of prior sets.
It would be useful to evaluate the proposed method on more complex models to further explore the theoretical and practical challenges associated with robust Bayesian analysis; as well as to investigate how other discrepancy measures such as those proposed by \citep{Martino_2017} and the Kullback-Leibler (KL) divergence measure could be used.

\section*{Supplementary material}

The Stan and R codes to run the analysis, except the data, are available through the following link: https://github.com/Iraices/IIS\_MCMC.  

\section*{Acknowledgement(s)}
We thank Claes Bernes for supporting this work with data on the meta-analysis in the evidence synthesis example. 

\section*{Disclosure statement}
No potential conflict of interest was reported by the author(s).

\section*{Funding}
This work was supported by the Swedish research council FORMAS through the project ``Scaling up uncertain environmental evidence'' (219-2013-1271) and the strategic research areas BECC (Biodiversity and  Ecosystem Services in a Changing Climate) and MERGE (Modelling the Regional and Global Climate/Earth system).

\bibliographystyle{plainnat}
\bibliography{references_IIS_MCMC}

\appendix
\section{Calculations}  \label{Calculations}

Explicit calculations for \cref{Eq:Var_g_tilde_back} are given here.
Recall \cref{Eq:Var_g_tilde_back}
\begin{align*}
\Var_q \Bigl(g(X)\Bigr) & = \Var_q\Bigl( (f(X)-\mu) w_p(X) \Bigr)
\end{align*}
We will use that, for any random variable $h(X)$,
\begin{align*}
    \mathrm{E}_q\Bigl(h(X)w_p(X)\Bigr)
    &=\int h(x)w_p(x)q(x)dx
    =\int h(x)\frac{c p(x)}{q(x)} q(x)dx \\
    &=c\int h(x)p(x)dx
    =c\mathrm{E}_p\Bigl(h(X)\Bigr),
\end{align*}
which gives
\begin{align}
\begin{split}
\Var_q \Bigl(h(X)w_p(X) \Bigr)
= c \mathrm{E}_p \Bigl( h^2(X) w_p(X) \Bigr)
- c^2 \mathrm{E}_p\Bigl( h(X) \Bigr)^2
\end{split}
\label{eq:appendix:Eq_hw}
\end{align}
Taking $h(X) = f(X)-\mu$ and noting that $\mathrm{E}_p \Bigl(h(X)\Bigr)=0$, the expression in \cref{Eq:Var_g_tilde_back} expands as follows
\begin{align*}
  \Var_q\Bigl(g(X)\Bigr) & = \Var_q\Bigl( (f(X)-\mu) w_p(X) \Bigr) =
  c \mathrm{E}_p\Bigl( \left(f(X)-\mu\right)^2 w_p(X) \Bigr)
\end{align*}
Applying a second order delta method to the last expectation at $\mathrm{E}_p\Bigl(w_p(X)\Bigr)$ and $\mathrm{E}_p \Bigl(h(X) \Bigr)$, gives (see appendix B for details)
\begin{align}
\Var_q\Bigl(g(X)\Bigr) & \approx c \Bigl(\mathrm{E}_p\Bigl(h(X)\Bigr)^2 \mathrm{E}_p \Bigl(w_p(X)\Bigr) + \mathrm{E}_p\Bigl(h(X)\Bigr) \mathrm{Cov}_p \Bigl(w_p(X),h(X)\Bigr) \nonumber \\
& +  \mathrm{E}_p \Bigl(w_p(X)\Bigr) \Var_p\Bigl(h(X)\Bigr)\Bigr) \label{eq: matching_appendixB} \\
& = c \mathrm{E}_p \Bigl(w_p(X)\Bigr) \mathrm{E}_p \Bigl(h^2(X)\Bigr)  \\
& =
c \mathrm{E}_p\Bigl( \left(f(X)-\mu\right)^2 \Bigr)  
\mathrm{E}_p\Bigl( w_p(X) \Bigr)
\\ &= 
c \Var_p\Bigl( f(X) \Bigr) \mathrm{E}_p\Bigl( w_p(X) \Bigr).
\label{eq:appendix:varqwpx}
\end{align}

Using \cref{eq:appendix:varqwpx} and the equality $\Var_p(f(X))=N \Var_p(\overline{\mu})$ we obtain
\begin{align}
\Var_q(g(X)) & \approx N \Var_p(\overline{\mu}) c \mathrm{E}_p\Bigl( w_p(X) \Bigr)
\\ & =
N \Var_p(\overline{\mu}) \mathrm{E}_q\Bigl( w^2_p(X) \Bigr)
\\ & =
N \Var_p(\overline{\mu}) \left(\frac{1}{N}\sum_{i=1}^N w_p^2(X_i)\right).
\end{align}
Here we have used that \cref{eq:appendix:Eq_hw} with $h(X)=w_p(X)$ gives the equality $c \mathrm{E}_p( w_p(X) ) = \mathrm{E}_q( w_p^2(X) )$.

\section{Delta method} 

Let $g(u,v)$ be a function twice differentiable at $(u,v) = (a_1,a_2)$. Then, the second order Taylor polynomial for $g(u,v)$ near the point $(u,v) = (a_1,a_2)$ is:
\begin{align}
g(u,v)_{(a_1,a_2)} & = g(a_1,a_2) + \frac{\partial g}{\partial u}(a_1,a_2)(u - a_1) + \frac{\partial g}{\partial v}(a_1,a_2)(v - a_2) \nonumber \\
& + \frac{1}{2} \frac{\partial^2 g}{\partial u^2}(a_1,a_2)(u - a_1)^2 + \frac{1}{2} \frac{\partial^2 g}{\partial u \partial v}(a_1,a_2)(u - a_1) (v - a_2) \nonumber \\
& + \frac{1}{2} \frac{\partial^2 g}{\partial v^2}(a_1,a_2)(v - a_2)^2.
\end{align}

Now, we can use the second order Taylor polynomial approximation to estimate the mean (this is also known as second order delta method).

Let $U$ and $V$ be random variables with mean $\theta_1 = \mathrm{E}_p(U)$ and $\theta_2 = \mathrm{E}_p(V)$ respectively.
\begin{align}
\mathrm{E}_p\Bigl(g(U,V)\Bigr)_{(\theta_1,\theta_2)} & \approx g(\theta_1,\theta_2) + \frac{\partial g}{\partial u}(\theta_1,\theta_2) \mathrm{E}_p(U - \theta_1) + \frac{\partial g}{\partial v}(\theta_1,\theta_2) \mathrm{E}_p(V - \theta_2) \nonumber \\
& +  \frac{1}{2} \frac{\partial^2 g}{\partial u^2}(\theta_1,\theta_2)\mathrm{E}_p\Bigl( (U - \theta_1)^2\Bigr)  \nonumber \\
& + \frac{1}{2} \frac{\partial^2 g}{\partial u \partial v}(\theta_1,\theta_2)\mathrm{E}_p \Bigl( (U - \theta_1) (V - \theta_2) \Bigr) \nonumber \\
& + \frac{1}{2} \frac{\partial^2 g}{\partial v^2}(\theta_1,\theta_2)\mathrm{E}_p \Bigl( (V - \theta_2)^2 \Bigr).
\end{align}

Note that $\mathrm{E}_p(U - \theta_1) = 0$ and $\mathrm{E}_p(V - \theta_2) = 0$. 
Taking $g(u,v)\coloneqq v^2 u$ where $V = h(X)$ and $U = w_p(X)$ (to match notation in \cref{eq: matching_appendixB})  
\begin{align}
\mathrm{E}_p\Bigl(g(U,V)\Bigr)_{(\theta_1,\theta_2)} & \approx g(\theta_1,\theta_2) + \theta_2 \mathrm{Cov}_p(U,V) +  \theta_1 \Var_p(V) \nonumber \\
 & \approx \Bigl(\mathrm{E}_p(V)\Bigr)^2 \mathrm{E}_p(U) + \mathrm{E}_p(V) \mathrm{Cov}_p(U,V) +  \mathrm{E}_p(U) \Var_p(V)   \\
 & \approx \Bigl(\mathrm{E}_p\Bigl(h(X)\Bigr)\Bigr)^2 \mathrm{E}_p\Bigl(w_p(X)\Bigr) + \mathrm{E}_p\Bigl(h(X)\Bigr) \mathrm{Cov}_p\Bigl(w_p(X),h(X)\Bigr)  \nonumber \\
 & +  \mathrm{E}_p(w_p(X)) \Var_p\Bigl(h(X)\Bigr). 
\end{align}

\section{Proof of existence of minimum} \label{proof_minimum}

To guarantee the minimum exist in our example, it is enough to prove that the expectation is 
continuous as a function of $(\mu_0, \tau_0)$ on $[-8,68] \times [5,16]$.

By Fubini's theorem in our example, we have 
\begin{equation}
	\mathrm{E}_{(\mu_0, \tau_0)}(\mu) \coloneqq \displaystyle
	\int_{1}^{\tau_0} \int_1^5 \int_{- \infty}^{+\infty}  \mu \cdot p_{(\mu_0, \tau_0)}(\mu,\tau_\mu,k) \,d\mu \,dk \,d\tau_\mu, \label{Eq:expectation}
\end{equation}
where 
\begin{align*}
	p_{(\mu_0, \tau_0)}(\mu,\tau_\mu,k) &= \frac{1}{c(\mu_0,\tau_0)} \cdot \Biggl(\prod_{i=1}^N  \frac{1}{\sqrt{\sigma^2_i + k^2\tau_\mu^2}}\Biggr)
	\exp  \Biggl\{ -\frac{1}{2} \sum_{i = 1}^N \frac{(y_i - \mu)^2}{\sigma^2_i + k^2\tau_\mu^2} \Biggr\} \cdot\\ \cdot
	& \frac{1}{ \tau_\mu} \exp  \Bigl\{ -\frac{1}{2}\frac{(\mu - \mu_0)^2}{\tau_\mu^2} \Bigr\}  \nonumber
\end{align*}
is the posterior probability distribution. The proportionality constant $c(\mu_0,\tau_0)$ is given by 
\begin{align}
\label{Eq: prop_constant}
	c(\mu_0,\tau_0)  & =  \int_{1}^{\tau_0} \int_1^5 
	\Biggl( \prod_{i=1}^N \frac{1}{\sqrt{\sigma^2_i + k^2\tau_\mu^2}}\Biggr)
	 \frac{1}{\tau_\mu} \\
	& \Biggl[ \int_{- \infty}^{+\infty} 
	\exp  \Bigl\{-\frac{1}{2} \sum_{i = 1}^N \frac{(y_i - \mu)^2}{\sigma^2_i + k^2\tau_\mu^2} \Bigr\} \cdot 
	 \exp  \Bigl\{ -\frac{1}{2}\frac{(\mu - \mu_0)^2}{\tau_\mu^2} \Bigr\} 
	  \,d\mu \Biggr ] \,dk  \,d\tau_\mu \nonumber \\ 
	&= \int_{1}^{\tau_0} \int_1^5 
	\Biggl( \prod_{i=1}^N \frac{1}{\sqrt{\sigma^2_i + k^2\tau_\mu^2}}\Biggr)
	 \frac{1}{ \tau_\mu} \\
	& \Biggl[ \int_{- \infty}^{+\infty} 
	\exp  \Biggl\{-\frac{1}{2} \Biggl( a(\tau_\mu, k, \mu_0)\mu^2 -2b(\tau_\mu, k, \mu_0) \mu + c(\tau_\mu, k, \mu_0)\Biggr)  \Biggr\} 
	  \,d\mu \Biggr ] \,dk  \,d\tau_\mu \nonumber   
\end{align}
where
\begin{align}
	a(\tau_\mu, k, \mu_0) & = \sum_{i =1}^N \frac{1}{\sigma_i^2 + k^2\tau_\mu^2} + \frac{1}{\tau_\mu^2}, \label{Eq:a} \\ 
	b(\tau_\mu, k, \mu_0) & = \sum_{i =1}^N \frac{y_i}{\sigma_i^2 + k^2\tau_\mu^2} + \frac{\mu_0}{\tau_\mu^2}, \label{Eq:b}\\
	c(\tau_\mu, k, \mu_0) & = \sum_{i =1}^N \frac{y_i^2}{\sigma_i^2 + k^2\tau_\mu^2} + \frac{\mu_0^2}{\tau_\mu^2} \label{Eq:c}.
\end{align}
Since $a(\tau_\mu, k, \mu_0) > 0$, then by Lemma \ref{lemma_1} 
in \eqref{Eq: prop_constant}, we get that the integral with respect to $\mu$ is  
\begin{equation}
	\sqrt{2\pi}\cdot a(\tau_\mu, k, \mu_0)^{\frac{-1}{2}} \cdot \exp \Biggl\{-\frac{1}{2} \Biggl(c(\tau_\mu, k, \mu_0) - \frac{b(\tau_\mu, k, \mu_0)^2}{a(\tau_\mu, k, \mu_0)}\Biggr)\Biggr\}.
\end{equation}

Note that $a(\tau_\mu, k, \mu_0)$, $b(\tau_\mu, k, \mu_0)$ and $c(\tau_\mu, k, \mu_0)$ are continuously differentiable and  $\tau_\mu \in [1, \tau_0]$ and $k \in [1,\, 5]$. 
Thus, by Leibniz integral rule it follows that
$c(\mu_0,\tau_0)$ is continuously differentiable on $[-8,68] \times [5,16]$. 

Analogously, applying Lemma \ref{lemma_2} 
to \cref{Eq:expectation}, we get that the integral with respect to $\mu$ is 
\begin{equation}
\sqrt{2\pi}\cdot a(\tau_\mu, k, \mu_0)^{\frac{-3}{2}} \cdot b(\tau_\mu, k, \mu_0) \cdot \exp \Biggl\{-\frac{1}{2} \Biggl (c(\tau_\mu, k, \mu_0) - \frac{b(\tau_\mu, k, \mu_0)^2}{a(\tau_\mu, k, \mu_0)} \Biggr) \Biggr\}.
\end{equation}
Furthermore, by Leibniz integral rule it yields that
$\mathrm{E}_{(\mu_0, \tau_0)}(\mu)$ is continuously differentiable on $[-8,68] \times [5,16]$.  

\section {Lemmas} 

\begin{lemma} \label{lemma_1}
Let $a > 0, b \in \mathbb{R}$ and $c \in \mathbb{R}$. It holds  that:
\begin{align}
	& \int_{-\infty}^{+\infty} \exp \Biggl\{-\frac{1}{2}(ax^2 -2bx + c)\Biggr\} \,dx  \\
	& = \sqrt{2\pi} \cdot a^{-\frac{1}{2}} \cdot \exp \Biggl \{  -\frac{1}{2} \Biggl(c - \frac{b^2}{a} \Biggr) \Biggr \}.
\end{align}
\end{lemma}

\begin{proof}
	By completing the square, we obtain 
	\begin{align}
		& \int_{-\infty}^{+\infty} \exp \Biggl\{-\frac{1}{2}\Biggl(\frac{x - \frac{b}{a}} {a^{-\frac{1}{2}}}\Biggr) ^2 \Biggr\}
		\cdot \exp \Biggl\{-\frac{1}{2}\Biggl(c - \frac{b^2}{a}\Biggr) \Biggr \}  \,dx \\ 
		& = \sqrt{2 \pi} \cdot a^{-\frac{1}{2}} \cdot \exp \Biggl\{-\frac{1}{2} \Biggl(c - \frac{b^2}{a}\Biggr)\Biggr\} \cdot 
		\int_{-\infty}^{+\infty} \frac{1}{\sqrt{2 \pi} \cdot a^{-\frac{1}{2}}}
		\exp \Biggl\{-\frac{1}{2} \Biggl(\frac{x - \frac{b}{a}} {a^{-\frac{1}{2}}}\Biggr) ^2 \Biggr\} \,dx.
	\end{align}

Note that the integrand is the probability density function of a normally distributed 
random variable with mean $\frac{b}{a}$ and variance $\frac{1}{a}$, (i.e. $N \Bigl (\frac{b}{a}, \frac{1}{a}\Bigr)$). 
Thus, the desired result immediately follows.
\end{proof}

\begin{lemma} \label{lemma_2}
Let $a > 0, b \in \mathbb{R}$ and $c \in \mathbb{R}$. It holds  that:
\begin{align}
	& \int_{-\infty}^{+\infty} x \exp \Biggl\{-\frac{1}{2}(ax^2 -2bx + c)\Biggr\} \,dx  \\
	& = \sqrt{2\pi} \cdot a^{-\frac{3}{2}}\cdot b \cdot \exp \Biggl \{ -\frac{1}{2}\Biggl(c - \frac{b^2}{a} \Biggr) \Biggr \}.
\end{align}
\end{lemma}

\begin{proof}
By completing the square, we obtain 
\begin{align}
	& \int_{-\infty}^{+\infty} x \cdot \exp \Biggl\{-\frac{1}{2}\Biggl(\frac{x - \frac{b}{a}} {a^{-\frac{1}{2}}}\Biggr) ^2 \Biggr\}
	\cdot \exp \Biggl\{-\frac{1}{2} \Biggl(c - \frac{b^2}{a} \Biggr) \Biggr \}  \,dx \\ 
	& = \sqrt{2 \pi} \cdot a^{-\frac{1}{2}} \cdot \exp \Biggl\{-\frac{1}{2} \Biggl(c - \frac{b^2}{a} \Biggr)\Biggr\} \cdot 
	\int_{-\infty}^{+\infty} \frac{1}{\sqrt{2 \pi} \cdot a^{-\frac{1}{2}}}
		x \exp \Biggl\{-\frac{1}{2} \Biggl(\frac{x - \frac{b}{a}} {a^{-\frac{1}{2}}}\Biggr) ^2 \Biggr\} \,dx.
\end{align}

Note that the previous integral is the expected value of a normally distributed random variable with 
mean $\frac{b}{a}$ and variance $\frac{1}{a}$, (i.e. $N \Bigl (\frac{b}{a}, \frac{1}{a}\Bigr)$). 
Thus, the desired result immediately follows.
\end{proof}

\end{document}